\documentclass[12pt,epsf]{article}
\usepackage{amsmath,amssymb}
\usepackage{amsthm} 
\usepackage{cite}
\usepackage{graphicx}

\newtheorem{theorem}{Theorem}[section]

\newtheorem{lemma}[theorem]{Lemma}

\DeclareMathOperator{\real}{Re}

\newcommand{\R}{\mathbb{R}}

\begin{document}
\begin{titlepage}
\begin{center}

 \vspace{-0.7in}

{\large \bf  Disordered Field Theory in $d=0$ \\
\vspace{.06in} and  \\
\vspace{.06in}
Distributional Zeta-Function}\\

\vspace{.5in}
{\large\em
B. F. Svaiter,\,\,\footnotemark[1]  N. F. Svaiter\,\,\footnotemark[2]}\\
\vspace{.08in}

Instituto de Matem\' atica Pura e Aplicada - IMPA \footnotemark[1] \\
Estrada Dona Castorina 110 Rio de Janeiro.
 RJ 22460-320, Brazil\\
\vspace{.06in}

Centro Brasileiro de Pesquisas F\'\i sicas - CBPF \footnotemark[2]\\
Rua Dr. Xavier Sigaud 150
Rio de Janeiro, RJ,22290-180, Brazil\\

\subsection*{\\Abstract}
\end{center}

Recently we introduced a new technique for computing 
the average free
energy of a system with quenched randomness.
The basic tool of this technique is
a distributional zeta-function. 
The distributional zeta-function is a complex function
whose derivative at the origin yields the average free energy of the system as the sum of two
contributions: the first one is
a series in which all the
integer moments of the partition function of the model contribute; the second one, which can not be
written as a series of the integer moments, can be made as small as desired.
In this paper we present a mathematical rigorous proof that the average free energy of one
disordered $\lambda\varphi^{4}$ model defined in
a zero-dimensional space can be obtained using the distributional zeta-function technique.
We obtain an analytic expression for the average free energy of the model.

\bigskip

\vspace{.06in}

\noindent
{\sc keywords:} disordered systems; average free-energy; replicas; distributional zeta-function.

\vspace{.06in}

\footnotetext[1]{e-mail:\,\,benar@impa.br}
\footnotetext[2]{e-mail:\,\,nfuxsvai@cbpf.br}
\vspace{.09in}
PACS numbers: 05.20.-y,\,75.10.Nr

\end{titlepage}
\newpage\baselineskip .18in

%
%
\quad

Disordered systems have been investigated for decades in statistical mechanics
\cite{binder,livro1,livro2,livro3,livro4}, 
gravitational physics \cite{pe2,pe10,pe11,dis1,dis2,dis3}, number theory \cite{jpa} and condensed matter. 
For the case of disordered systems with quenched disorder, one is mainly interested in
averaging the free energy over the disorder, which amounts to averaging the log of the partition function $Z$.
The central problem is how to average the disorder dependent free energy over the ensemble of configurations of the disorder
degrees of freedom. The main approaches to obtain the average free energy are the dynamical approach \cite{dynamical},
the cavity method \cite{cavity,livromez} and also the replica method \cite{edwards}. 
In the replica method, the expected value of the partition function's $k$-th power $Z_k=\mathbb{E}\, Z^k$ is 
computed by integrating over the disorder field and the average free energy is obtained using the formula
 $\mathbb{E}\,{\ln Z}=
\lim_{k\rightarrow 0}\frac{Z_{k}-1}{k}$, where $Z_{k}$ for $0<k<1$ is derived from its values for $k$ integer.
This approach sometimes requires a  symmetry breaking procedure to yield physical sound results \cite{pa1,pa2,pa3}.

Despite the success in the application of the replica method in disordered systems,
some authors consider that a mathematical rigorous derivation to support this procedure is still lacking
\cite{cri1,cri2,cri4,cri3}. 
It is therefore natural
to ask whether there exists
a mathematically rigorous method, based on the use of replicas, for computing the
average free energy of systems with quenched disorder.
In Ref. \cite{sum}, Dotsenko considered an alternative approach where the summation of
all integer moments of the partition function is used
to evaluate the average free energy of the random energy model. Also 
a replica calculation using only the integer moments of the
partition function have been considered in Ref. \cite{virasoro}.
In this paper we present a mathematical rigorous use of a new procedure to find the average free energy   
in systems with quenched disorder \cite{distributional}. 

We associate with systems with quenched disorder a complex function 
which, 
due to its similarities with zeta-functions, we call distributional zeta-function,
obtained by an integral with respect to a probability distribution.
We will show that the derivative of the distributional zeta-function at the origin
yields the
average free energy of the underlying system with quenched disorder.
For simplicity we study a disordered zero-dimensional field theory model, which allows a rigorous 
use of our method. Nevertheless, all the computations can be formally extended to more complex physical models where
the support of the  disorder probability distributions are infinite-dimensional spaces. 
In our derivation the average free energy of the zero-dimensional disordered
$\lambda\varphi^{4}$ model is given by the sum of two terms, a series in which all the integer 
moments of the partition function contribute and  a term that can not be
written as a series of the integer moments, but can be made as small as desired.

We would like to point out that zero-dimensional models has been widely studied in field theory \cite{argyres,rivasseau}. For instance,
it is well known that many perturbative series in quantum field theory has zero radius of convergence \cite{dyson,lipatov}. Although this zero radius of
convergence is encountered in many models in quantum field theory models in a four-dimensional space-time,
the nature of the perturbative expansion has been investigated also in systems with disorder is smaller dimensions. 
The structure of the perturbative
expansion of system without \cite{sce2,sce1,sce3} or with quenched disorder in models in zero dimensions was investigated in many papers 
\cite{sum1,sum2}. Zero-dimensional models appear in the strong-coupling expansion in field theory
\cite{s1,s2,s3,s4}, and also is closed related with matrix models since gauge theories in zero-dimensions are described by matrix models \cite{thooft, brezin}.  

The organization of this paper is as follows.
In section II we
introduce the distributional zeta-function
and compute the derivative of the distributional zeta-function at the origin in order to obtain the
average free energy of the system with quenched disorder. 
Conclusions are given in section III. 
We use $\hbar=c=k_{B}=1$.

\section{The distributional zeta-function in disordered models}
\quad

In the zero-dimensional scalar $\lambda\varphi^4$ model, without disorder, the action and the 
partition function are, respectively,
\begin{equation}
S(\varphi)=\frac{1}{2}
 m_{0}^{2}\varphi^2
+\frac{\lambda}{4!}\varphi^{4}\;\;\;\text{ and }\;\;\; Z=\int d\varphi\exp(-S(\varphi)).
\label{dis1}
\end{equation}
%
%
The normalized correlation functions of this model are just weighted integrals of kind
\begin{equation*}
\langle\,f\,\rangle=\frac{1}{Z}\int\,d\varphi\,f(\varphi)\,\exp(-S(\varphi)).
\end{equation*}
In the presence of a disorder $h$ linearly coupled with $\varphi$,
the action and partition function becomes
\begin{equation}
S(h,\varphi)=S(\varphi)+h\varphi,\qquad
\end{equation}
and
\begin{equation}
Z(h)=\int d\varphi\exp\left(-\frac{1}{2}
 m_{0}^{2}\varphi^2
-\frac{\lambda}{4!}\varphi^{4}-h\varphi\right).
\label{dis1}
\end{equation}

The disorder dependent free energy i.e., the $h$-dependent free energy $F(h)$ is given by
\begin{equation}
F(h)=\ln\,Z(h).
\label{fe}
\end{equation}
Let $\mu$ be the  probability distribution of the disorder, that is,
$\mu$ is a Borel measure in $\mathbb{R}$ and $\mu(\R)=1$. 
If
$\mu$ happens to have a probability density function $P(h)$, the probability distribution can be written as 
\[
d\mu(h)=dh P(h).
\]
A widely used distribution is the normal distribution
\begin{align*}
d\mu=dh\dfrac{1}{\sqrt{2\pi\sigma}}\exp\left(-\dfrac{(h-m)^2}{2\sigma}\right),
\end{align*}
nevertheless, we can consider more general probabilities, as for example, discrete ones.  
The average free energy $F_{q}$ is defined as
\begin{equation}
F_{q}=\int\,d\mu(h)\,F(h)=\int\,d\mu(h)\,\ln Z(h).
\label{sa27}
\end{equation}
To evaluate this integral, we will resort to a zeta-function method, 
extending to this context 
a useful  procedure of quantum field theory \cite{elizalde}.

Recall that a measure space $(\Omega,\mathcal{W},\eta)$ consist in a set $\Omega$,
a $\sigma$-algebra $\mathcal{W}$ in $\Omega$, and a measure $\eta$ in this
$\sigma$-algebra.
Given a measure space $(\Omega,\mathcal{W},\eta)$
and a measurable $f:\Omega\to(0,\infty)$, we define the associated  generalized $\zeta$-function as
\begin{equation*}
\zeta_{\,\eta,f}(s)=\int_\Omega f(\omega)^{-s}\, d\eta(\omega)
\end{equation*}
for those $s\in\mathbb{C}$ such that  $f^{-s}\in L^1(\eta)$,
where in the above integral $f^{-s}=\exp(-s\log(f))$ is obtained using the
principal branch of the logarithm.
This formalism encompasses some well-known instances of zeta-functions:
\begin{enumerate}
\item if $\Omega=\R_+$, $\mathcal{W}$ is the Lebesgue $\sigma$-algebra,  $\eta$  is the Lebesgue measure,
and $f(\omega)=\lfloor \omega\rfloor$ we retrieve the classical Riemann zeta-function \cite{riem,riem2};
\item if $\Omega$ and $\mathcal{W}$ are as in item 1, $\eta(E)$ counts the prime numbers in $E$ and $f(\omega)=\omega$
we retrieve  the prime zeta-function \cite{landau,fro,primes0,primes};
\item if $\Omega$, $\mathcal{W}$, and $f$ are as in item 2 and $\eta(E)$ counts the non-trivial zeros of the Riemann zeta-function, with
their respective multiplicity, we obtain the families of superzeta-functions \cite{voros}.
\item if $\Omega$, $\mathcal{W}$, and $f$ are as in item 2 and $\eta(E)$ counts the eigenvalues of an elliptic operator, with
their respective multiplicity, we obtain the spectral zeta-function \cite{seeley,ray,hawking,dowker,fulling}.

\end{enumerate}
Further extending this formalism to the case where $f(h)=Z(h)$ and $\eta=\mu$,  the probability distribution of $h$,
leads
to the definition of the distributional zeta-function $\Phi(s)$,
\begin{equation}
\Phi(s)=\int d\mu(h)\frac{1}{Z(h)^{s}}
\label{pro1}
\end{equation}
for $s\in \mathbb{C}$, this function being defined in the region where the
above integral converges. Before continue, we would like to point out that the free energy of the system with annealed disorder is given by
$$F_{a}=-\ln\,\Phi(s)|_{s=-1}.$$ 
The next technical lemma will be used to study the domain of definition of the distributional zeta-function.

\begin{lemma}\label{lm:t0}
For any $h\in \mathbb{R}$,  $Z(h) \geq Z(0)>0$.
\end{lemma}

\begin{proof}
In view of Eq. \eqref{dis1}, only the first inequality needs to be proved.
Take $h\in\mathbb{R}$. 
Since 
$Z(h)=Z(-h)=(Z(h)+Z(-h))/2$, 
again in view of Eq. \eqref{dis1}
we can write
\begin{equation*}
  Z(h)=\int d\varphi\,\, \exp\left(-\frac{1}{2}
 m_{0}^{2}\varphi^2
-\frac{\lambda}{4!}\varphi^{4}\right) \;
  \cosh\left(
h\varphi\right),
\end{equation*}
which trivially implies the desired inequality.
\end{proof}

Now we will prove, without further assumptions on $\mu$, that
$\Phi(s)$ is well defined in the half complex plane $\real(s) \geq 0$.

\begin{theorem}
The distributional zeta-function  $\Phi(s)$ specified in Eq. \eqref{pro1} is well defined and continuous for $\real(s) \geq 0$. 
\end{theorem}

\begin{proof}
It follows from Lemma~\ref{lm:t0} that for $\real(s) \geq 0$
\begin{equation*}
  \int d\mu(h)\,\left|\frac{1}{Z(h)^{s}}\right|\leq\int d\mu(h)
  \frac{1}{Z(0)^{\real(s)}}=
   \dfrac{1}{Z(0)^{\real(s)}}<\infty.
\end{equation*}
Therefore, the integral in Eq.\ \eqref{pro1} is convergent in the half
complex plane $\real(s)\geq 0$
and $\Phi$ is well defined in this region, \emph{without} resorting to analytic continuations.
Take  $M>0$ and define $C_M=\max\{1,1/Z(0)^M\}$. 
Using again Lemma~\ref{lm:t0} 
we conclude that
\begin{equation*}
\left|\dfrac{1}{Z(h)^s}\right| \leq C_M \;\;\;\text{ for }\;\;\; 0\leq \real(s)\leq M.
\end{equation*} 
Since $s\mapsto 1/Z(h)^s$ is continuous (for each $h$) and $\int d\mu(h)\,C_M=C_M<\infty$,
it follows from the above inequality and Lebesgue's dominated convergence theorem that
$\Phi(s)$ is continuous in the strip $0\leq \real(s)\leq M$. As $M>0$ is arbitrary, $\Phi(s)$
is continuous for $\real(s)\geq 0$.
\end{proof}

Next we show that the average fee energy can be retrieved from $\Phi$.
%
%
%
From now on 
\[
 \left.\dfrac{d}{ds}f(s)\right|_{s=0^+}=\lim_{s\to 0^+}
 \frac{f(s)-f(0)}{s}
\]
whenever this limit exists.

\begin{theorem}
\label{th1}
If $F_q$, the average free energy defined by Eq. \eqref{sa27}, is well defined,
then $(d/ds)\Phi(s)|_{s=0^+}$ exists and
\begin{align}
F_{q}=-\int d\mu(h)\left.\frac{d}{ds}\frac{1}{Z(h)^{s}}\right|_{s=0^{\,+}}
=-\left.\frac{d}{ds}\Phi(s)\right|_{s=0^{\,+}}.
\label{sa22}
\end{align}
\end{theorem}

\begin{proof}
Only the second equality needs to be proved.
To prove it,
take $1>s>0$ and $h\in\mathbb{R}$. 
It follows from the mean value theorem that there is $0<\theta<1$ such that
\begin{align*}
\left|\dfrac{Z(h)^{-s}-1}{s}\right|=\left|Z(h)^{-\theta s}\ln Z(h) \right|.
\end{align*}
Direct use of the inequality
$Z(h) \geq Z(0)>0$ yields
\begin{align*}
\left|Z(h)^{-\theta s}\ln Z(h) \right|
=\left|{Z(0)}^{-\theta s}\ln Z(h)\right|
\leq|\ln Z(h)|\max\{1,Z(0)^{-1}\}.
\end{align*}
Hence, 
\begin{align*}
\left|\dfrac{Z(h)^{-s}-1}{s}\right|\leq |\ln Z(h)|\max\{1,Z(0)^{-1}\}\in L^1(\mu)
\;\;\text{ for } 0<s<1,
\end{align*}
where the inclusion follows from the assumption of the integral 
on Eq. \eqref{sa27} being well defined.
It follows from the above equation and Lebesgue dominated convergence theorem that
\begin{equation*}
\lim_{s\to 0^+}\int d\mu(h)\dfrac{Z(h)^{-s}-Z(h)^0}{s}
=\int d\mu(h)\lim_{s\to 0^+}\dfrac{Z(h)^{-s}-Z(h)^0}{s}
\end{equation*}
and the conclusion follows.
\end{proof}

From now on we assume that the average free energy as specified in Eq. \eqref{sa27}
is well defined, so that this physical quantity can be obtained by Eq.
\eqref{sa22},
%
Note that Eqs. \eqref{pro1} and  \eqref{sa22} provides an analytic expression
for  $F_q$ which \emph{does not require derivation
of the (integer) moments of the partition function}. 

To obtain from Eq. \eqref{sa22} a new expression for the average free
energy, we will derive another integral representation for the
distributional zeta-function. Direct use of Euler's integral
representation for the gamma function give us
\begin{equation*}
\frac{1}{Z(h)^{s}}=\frac{1}{\Gamma(s)}
\int_{0}^{\infty}dt\,t^{s-1}e^{-Z(h)t},\,\,\,\,\, \text{for}\,\,\,\, \real\,(s)>0.
\end{equation*}
%
%
Substituting the above equation in Eq.\ \eqref{pro1} we get
\begin{equation}
\Phi(s)=\dfrac{1}{\Gamma(s)}\int d\mu(h)\int_0^\infty dt\,t^{s-1}e^{-Z(h)t}, \;\;\text{ for }\;\;
\real(s)>0.
\label{pro1.b}
\end{equation}

Recall that the integer moments of the partition function are 
\begin{align*}
\mathbb{E}\,Z^k=\int d\mu(h)\, Z(h)^k\qquad k=1,2,\dots
\end{align*}
The family of integer moments of $Z$ are called in the literature the replica partition function.
If the probability distribution $\mu$ has compact support, which is to say 
that $\mu\left(\mathbb{R}\setminus[-r,r]\right)=0$ for $r$ large enough, then
there is a geometric bound for the growth of the partition function
moments, as proved in the next lemma.
\newpage

\begin{lemma}
\label{lm:bm}
If $\mu$ has compact support, then there exists $\alpha,\beta>0$ such that
$\mathbb{E}\,Z^k \leq \alpha\,\beta^k$
for any $k$.
\end{lemma}

\begin{proof}
Choose $k\geq 1$.
Using the notation
\begin{align}
  \label{eq:notations}
 \vec \varphi=(\varphi_1,\dots,\varphi_k),\quad
 \left\|\vec\varphi\right\|_p=(|\varphi_1|^p+\dots+|\varphi_k|^p)^{1/p},
  \quad
  \mathbf{1}_k=(1,\dots,1)\in\mathbb{R}^k
\end{align}
we can write
\begin{align}
  \label{eq:mkbound}
  \int d\mu(h)Z(h)^k=\int d\mu(h)\int\prod_{i=1}^kd\varphi_i
  \exp\left(-\dfrac{m_0}{2}\|\vec\varphi\|_2^2-
  \dfrac{\lambda}{4!}\|\vec\varphi\|_4^4-h\langle\mathbf{1}_k,\vec\varphi
  \rangle
  \right).
\end{align}
It follows from H\"older inequality
for the conjugate
exponents $p=4$ and $q=4/3$, that
\begin{align*}
 | \langle\mathbf{1}_k,\vec\varphi
  \rangle|\leq k^{3/4} \|\vec\varphi\|_4.
\end{align*}
For any $\tau\in\mathbb{R}$,
\[
-\dfrac{\lambda}{4!}\tau^4+|h\tau|k^{3/4}\leq 
\dfrac{3}{4}\left(\dfrac{3!}{\lambda}\right)^{1/3}k|h|^{4/3}.
\]
It follows from the two above inequalities and Eq. \eqref{eq:mkbound} that
\begin{align*}
    \int d\mu (h) Z(h)^k&\leq\int d\mu(h)\int \prod_{i=1}^kd\varphi_i
    \exp\left(-\dfrac{m_0}{2}\|\vec\varphi\|_2^2+kC_\lambda|h|^{4/3}\right)\;\text{ for }\;
   C_\lambda=\dfrac{3}{4}\left(\dfrac{3!}{\lambda}\right)^{1/3}.
\end{align*}
For $r>0$ large enough the interval $[-r,r]$ contains the support of $\mu$ and, therefore,
\begin{align*}
    \int d\mu (h) Z(h)^k&\leq
		\int_{[-r,\,r]} d\mu(h)\int \prod_{i=1}^kd\varphi_i
    \exp\left(-\dfrac{m_0}{2}\|\vec\varphi\|_2^2+kC_\lambda|h|^{4/3}\right)\\&\leq
		\exp(kC_\lambda r^{4/3})\int \prod_{i=1}^kd\varphi_i
    \exp\left(-\dfrac{m_0}{2}\|\vec\varphi\|_2^2\right)
\end{align*}
which proves the lemma.
\end{proof}

Our aim now is to use the representation of $\Phi$ provided in Eq. \eqref{pro1.b} to express
the average free energy of the system as the sum of two
contributions: the first one is
a series in which all the
integer moments of the partition function of the model contribute; the second one, which can not be
written as a series of the integer moments, can be made as small as desired.
We will show that such a representation can be obtained whenever the probability distribution $\mu$, 
has compact support. 

\begin{theorem}
\label{th:main}
If $\mu$ has compact support, then for any $a>0$,
\begin{equation}
F_q=\sum_{k=1}^\infty \frac{(-1)^{k+1}a^{k}}{k!k}\,
\mathbb{E}\,{Z^{\,k}}
-\bigr(\ln(a)+\gamma\bigl)+R(a),
\label{m23e}
\end{equation}
where $\gamma=0.577\dots$ is Euler's constant and 
\begin{equation}
R(a)=-\int d\mu(h)\int_{a}^{\infty}\,\dfrac{dt}{t}\, e^{-Z(h)t},\qquad
|R(a)| \leq \dfrac{1}{Z(0)a}\exp\big(-Z(0)a\big).
\label{m24}
\end{equation}
\end{theorem}

\begin{proof}[Proof of Theorem~\ref{th:main}]
take $a>0$ and write $\Phi=\Phi_1+\Phi_2$ where
\begin{align}
\label{m24}
\begin{aligned}
\Phi_{1}(s;a)&=\frac{1}{\Gamma(s)}\int d\mu(h)\int_{0}^{a}\,dt\, t^{s-1}e^{-Z(h)t},
\\
\Phi_{2}(s;a)&=\frac{1}{\Gamma(s)}\int d\mu(h)\int_{a}^{\infty}\,dt\, t^{s-1}e^{-Z(h)t}.
\end{aligned}
\end{align}
In view of Lemma~\ref{lm:t0}, the integral $\Phi_{2}(s;a)$ defines a 
an analytic function in $s$ on the whole complex plane.
Moreover,
\begin{align*}
  \begin{aligned}
    \dfrac{d}{ds}\Phi_2(s;a)
    &=\left(\dfrac{d}{ds}\frac{1}{\Gamma(s)}\right)\int d\mu(h)\,
    \int_{a}^{\infty}\,dt\, t^{s-1}e^{-Z(h)t}
    \\
    &\quad+\frac{1}{\Gamma(s)}\int
    d\mu(h)\int_{a}^{\infty}\,dt\, t^{s-1}\ln t\,e^{-Z(h)t}.
  \end{aligned}
\end{align*}
Since  $\Gamma(s)$ has a first-order pole at $s=0$ with residue $1$, 
\begin{equation}
-\dfrac{d}{ds}\Phi_{2}(s;a)|_{s=0}=-\int d\mu(h)\int_{a}^{\infty}\,\dfrac{dt}{t}\, e^{-Z(h)t}
=R(a).
\label{m24.1}
\end{equation}
Direct use of the definition of $R(a)$ and of  Lemma~\ref{lm:t0} yields the bound
\begin{equation}
\left|R(a)\right| \leq
\int d\mu(h)\int_{a}^{\infty}\,\dfrac{dt}{t}\, e^{-Z(0)t}
\leq\dfrac{1}{Z(0)a}\exp\big(-Z(0)a\big).
\label{m24.b}
\end{equation}

In the innermost integral in $\Phi_{1}(s;a)$ the series representation for
the exponential converges uniformly (for each $h$). Since the domain of
this integral is bounded, we can interchange the order of this integration with the
summation of the series to obtain
\begin{align}
  \begin{aligned}
    \Phi_{1}(s;a)&=\int d\mu(h)\frac{1}{\Gamma(s)}\sum_{k=0}^\infty
    \frac{(-1)^{k}a^{k+s}}{k!(k+s)}Z(h)^k\\
   &=\frac{a^s}{\Gamma(s+1)}+
\frac{1}{\Gamma(s)}\int d\mu(h)\sum_{k=1}^\infty
    \frac{(-1)^{k}a^{k+s}}{k!(k+s)}Z(h)^k.
  \end{aligned}
\label{m23b}
\end{align}
where the second equality is obtained by direct integration of
the term $k=0$ and the use of the identity 
 $\Gamma(s)s=\Gamma(s+1)$.
%
We claim that for $\real(s)>-1$
\begin{align}
\sum_{k=1}^\infty\int d\mu(h)\sum_{k=1}^\infty
    \left|\frac{a^{k+s}}{k!(k+s)}Z(h)^k\right|=\sum_{k=1}^\infty
    \frac{|a^{k+s}|}{k!|k+s|}\int d\mu(h)Z(h)^k<\infty.
\label{m23bb}
\end{align}
The equality follows from the monotone convergence theorem
applied to the partial sums of the series of \emph{nonegative} functions in the
first integral while the inequality follows from Lemma~\ref{lm:bm}.
In view of the above inequality, we can interchange the order of integration and
sum on the right hand-side of the last equality of Eq. \eqref{m23b} to obtain
\begin{align}
  \begin{aligned}
    \Phi_{1}(s;a)&=\frac{a^s}{\Gamma(s+1)}+
\frac{1}{\Gamma(s)}\sum_{k=1}^\infty
    \frac{(-1)^{k}a^{k+s}}{k!(k+s)}\mathbb{E}\,Z^{k}.
  \end{aligned}
\label{m23b1}
\end{align}
It follows from Lemma~\ref{lm:bm} that the series in the above equation is analytic
for $\real(s)>-1$. Therefore, using again the fact that $\Gamma(s)$ has a first
order pole in $0$ with residue 1 we conclude that
\begin{equation}
-\dfrac{d}{ds}\Phi_{1}(s)|_{s=0^{\,+}}=\sum_{k=1}^\infty \frac{(-1)^{k+1}a^{k}}{k!k}\,
\mathbb{E}\,{Z^{\,k}}
+f(a),
\label{m23d}
\end{equation}
where
\begin{equation}
f(a)=-\dfrac{d}{ds}\left(\dfrac{a^s}{\Gamma(s+1)}\right)|_{s=0}
=-\bigl(\ln(a)+\gamma\bigr)
\label{m23e}
\end{equation}
and $\gamma$ is Euler's constant $ 0.577\dots$
To end the proof, combine Eqs. \eqref{m24}, \eqref{m24.1},  \eqref{m23d}, and \eqref{m23e} with Theorem~\ref{th1}
and use the bound provided by Eq. \eqref{m24.b}.
\end{proof}

A case of special interest is $a=1$ in  Theorem~\ref{th:main};  with this choice
the average free energy
can be written as
\begin{equation}
F_q=\sum_{k=1}^\infty \frac{(-1)^{k+1}}{k!k}\,
\mathbb{E}\,{Z^{\,k}}
-\gamma+R(1),
\qquad \left|R(1)\right|\leq \dfrac{1}{Z(0)}\exp\big(-Z(0)\big).
\label{m23f}
\end{equation}
Observe that in Theorem~\ref{th:main} one cannot take the limit $a\to\infty$, because in this case the
series in Eq. \eqref{m23e} become meaningless.
Nevertheless, for $Z(0)$ bounded away from zero the contribution of $R(a)$ to the free energy can be made
as small as desired, taking $a$ large enough.
%

Note that the representation of the average free energy just by a series on the integer moments of the  partition function
would not describe the behavior of the free energy 
when $Z(h)\to 0$; the divergence of the free energy in this case comes from $R(a)=-(d/ds)\Phi_2(s)$
as revealed by direct inspection of Eq. (\ref{m24}). The contribution to the free energy due to the series expansion
captures its non-analytic behaviour when $Z(h)\to\infty$.
Note that the Eq. \eqref{m23b} does not require compactness of $\mu$'s support.


\section{Conclusions}

\quad

There is a growing interest in disordered systems in physics and many areas beyond physics.
The usual approach to study such disordered systems is to transform the random problem into a
translational invariant one. For quenched disorder, one is mainly interested in
averaging the free energy over the disorder, which amounts to averaging the log of the partition function $Z$,
the connected vacuum to vacuum diagrams.

The replica method is a powerful tool used to calculate the free energy of systems with quenched disorder.
Despite the absence of a mathematically rigorous derivation, the standard replica method provides correct results in many situations.
It is  natural to ask if it
is possible to find a mathematically rigorous derivation which legitimates the use of the replica partition functions for
computing the average free energy of the system. 

In this paper we use the distributional zeta-function technique to obtain a
a mathematically rigorous derivation of the average free energy of the zero-dimensional $\lambda\varphi^{4}$ model. Contrary to the standard
replica method, our method neither involves derivation of the integer moments of the partition function with respect to those indices, nor 
extension of these derivatives to non-integers values.
The derivative of the distributional zeta-function at $s=0$ yields the average free energy.
Making use of the Mellin transform
and analytic continuation,
it is possible to obtain a
series representation for the average free energy where all the
integer moments of the partition function of the model contribute. 
The average free energy of the system is the sum of two
contributions: the first one is
a series in which all the
integer moments of the partition function of the model contribute; the second one, which can not be
written as a series of the integer moments, can be made as small as desired.

\section{Acknowlegements}

We would like to acknowledge G. Krein, T. Miclitz and
G. Menezes for the fruitful discussions.
This paper was supported by Conselho Nacional de
Desenvolvimento Cientifico e Tecnol{\'o}gico do Brazil (CNPq).

\end{document}